\documentclass[11pt,a4paper]{article}
\usepackage[top=1in, bottom=1in, left=0.5in, right=0.5in]{geometry}
\usepackage[english]{babel}
\usepackage{amsthm,amssymb,mathrsfs,setspace,pstricks,booktabs,mathtools,amsmath,amsfonts,pdflscape,array,upgreek,rotating,graphicx,caption,esint,subcaption,tgcursor,tcolorbox,tikz,stackengine}
\usepackage[english]{babel}
\usepackage[T1]{fontenc}
\usepackage[utf8]{inputenc}

\usetikzlibrary{shapes,arrows}
\usepackage{multirow}
\newtheorem{definition}{Definition}
\usepackage{mathptmx}
\usepackage[title]{appendix}

\newtheorem{corollary}{Corollary}[section]
\newtheorem{theorem}{Theorem}[section]
\newtheorem{lemma}{Lemma}[section]
\newtheorem{remark}{Remark}[section]

\newtheorem{example}{Example}[section]
\numberwithin{equation}{section}
\usepackage{natbib}
\begin{document}
\setcounter{page}{1}

	\thispagestyle{empty}
	\markboth{}{}

	\pagestyle{myheadings}
	\markboth{}{ }
	
	\date{}

	\noindent  
	
	\vspace{.1in}
	
	{\baselineskip 25truept
		
		\begin{center}
			{\huge {\bf A goodness-of-fit test for testing exponentiality based on normalized dynamic survival extropy}}
            \footnote{\noindent{\bf $^{\#}$} E-mail: nitin.gupta@maths.iitkgp.ac.in,\\
				{\bf * }  corresponding author E-mail: gauravk@kgpian.iitkgp.ac.in}\\
			
		\end{center}
		
		\vspace{.07in}
		
		\begin{center}
			{\large {\bf Gaurav Kandpal$^{*}$ and Nitin Gupta $^{\#}$}}\\
			{\large {\it Department of Mathematics, Indian Institute of Technology Kharagpur, West Bengal 721302, India }}
			\\
		\end{center}
	}
 
\abstract{
The cumulative residual extropy (CRJ) is a measure of uncertainty that serves as an alternative to extropy. It replaces the probability density function with the survival function in the expression of extropy. This work introduces a new concept called normalized dynamic survival extropy (NDSE), a dynamic variation of CRJ. We observe that NDSE is equivalent to CRJ of the random variable of interest $X_{[t]}$ in the age replacement model at a fixed time $t$. Additionally, we have demonstrated that NDSE remains constant exclusively for exponential distribution at any time. We categorize two classes, INDSE and DNDSE, based on their increasing and decreasing NDSE values. Next, we present a non-parametric test to assess whether a distribution follows an exponential pattern against INDSE. We derive the exact and asymptotic distribution for the test statistic $\widehat{\Delta}^*$. Additionally, a test for asymptotic behavior is presented in the paper for right censoring data. Finally, we determine the critical values and power of our exact test through simulation. The simulation demonstrates that the suggested test is easy to compute and has significant statistical power, even with small sample sizes. We also conduct a power comparison analysis among other tests, which shows better power for the proposed test against other alternatives mentioned in this paper. Some numerical real-life examples validating the test are also included.}\\


\section{Introduction}
\cite{lad2015extropy} offered a resolution to the axiomatization of information theories based on entropy, which had been a persistent concern since \cite{shannon1948mathematical} first proposed it and was later followed by \cite{jaynes1957information}. \cite{lad2015extropy} demonstrated that Shannon's entropy function possesses a complementary dual function known as "extropy". Let $X$ be a non-negative random variable with the probability density function $f(x)$, the cumulative distribution function $F(x)$ and the survival function $\bar F(x)$. The extropy of $X$ is defined as
\begin{equation}
    J(X)=-\frac{1}{2}\int_{0}^{\infty} f^2(x)\mathrm{d}x.
\end{equation}
In addition, they examined the extropy function for densities and demonstrated that relative extropy is a counterpart to the Kullback-Leibler divergence, commonly acknowledged as a measure of continuous entropy. One can refer to \cite{shannon1948mathematical}, \cite{jaynes1957information} and \cite{lad2015extropy} to study more about it. \cite{jahanshahi2020cumulative} introduced CRJ of random variable $X$ as 
\begin{equation}
    \xi J(X)=-\frac{1}{2}\int_{0}^{\infty}\bar{F}^2(x)\mathrm{d}x.
\end{equation}

They showed that if for some $p>0.5$, $\mathbb{E}(X^p)<\infty$, then $\xi J(X)\in (-\infty,0]$. They studied the estimation and applications of $\xi J(X)$. They also provided some inequalities involving extropy and entropy.

\cite{abdul2021dynamic} proposed an alternative approach to assess the remaining uncertainty in lifetime random variables. They introduce a dynamic form of CRJ, which they term dynamic survival extropy. This measure is defined by
\begin{equation}
    J_s(X; t) = -\frac{1}{2}\int_t^\infty \left(\frac{\bar{F}(x) }{\bar{F}(t)}\right)^2 \, \mathrm{d}x.
\end{equation}
The dynamic survival extropy of $X$ is, in fact, the CRJ of the random variable $[X-t|X>t]$. One can refer to \cite{abdul2021dynamic} to know more about this. \cite{nair2020dynamic} proposed a dynamic cumulative extropy for the past lifetime also, defined by
\begin{equation}
    \bar{J}_s(X; t) = -\frac{1}{2}\int_0^t \left(\frac{{F}(x) }{{F}(t)}\right)^2 \, \mathrm{d}x.
\end{equation}
By analogy to \cite{abdul2021dynamic} and \cite{nair2020dynamic}, we proposed the following definition of the normalized dynamic survival extropy (NDSE).
\begin{definition}
    Let $X$ be a non-negative random variable with survival function $\bar{F}$, then NDSE is defined as
\begin{align}
    \eta(t)=-\frac{1}{2}\frac{\int_{0}^{t} \bar{F}^2(u)\mathrm{d}u}{1-\bar{F}^2(t)},
\end{align}
which is a normalized dynamic variation of CRJ.
\end{definition}

NDSE relates to the reliability theory of a well-known age replacement model. Age replacement is a frequently employed preventive maintenance strategy aimed at averting the failure of an item during operation, which can often be expensive or hazardous. The age replacement policy involves replacing an object with a new one either when it fails or when it reaches a predetermined age $t$, whichever happens first. Let $F$ denotes the probability distribution of the lifetime of a newly introduced item and $X_{[t]}$ represents the associated random variable of interest, following the age replacement strategy with a replacement time of $t$. \cite{barlow1996mathematical} provided the survival function of $X_{[t]}$, denoted by $R_t$, as
\begin{equation*}
   R_t(x)= \sum_{n=0}^{\infty}\bar{F}^n(t)\bar{F}(x-nt)I_{[nt,(n+1)t]}(x),\,\,\,\,x>0,
    \label{survival_Age}
\end{equation*}
where 
\begin{eqnarray*}
     I_{[a,b]}(x) &=& \left\{\begin{array}{ccc}
                   1, && x\in [a,b] \\ \\
                   0, && \mbox{otherwise}.
             \end{array} \right.
\end{eqnarray*}
We have
\begin{align}
   \xi J(R_{t})&=-\frac{1}{2}\int_{0}^{\infty} R_{t}^2(x)\mathrm{d}x\nonumber\\
   &=-\frac{1}{2}\frac{\int_{0}^{t} \bar{F}^2(u)\mathrm{d}u}{1-\bar{F}^2(t)}\\
 &=\eta(t).\nonumber
\end{align}
Hence, the CRJ for $X_{[t]}$ is the same as NDSE. If the average waiting time between consecutive failures is a crucial criterion in deciding whether we should opt for an age replacement policy, then a reasonable way to decide would be to test whether the distribution of the given system is exponential. Some work in this direction has been done by \cite{kayid2013further} and \cite{kattumannil2019simple} using mean time to failure (MTTF).

In the past few years, many other researchers also worked to provide suitable goodness-of-fit test for testing exponentiality. We refer to \cite{zardasht2015empirical}, \cite{baratpour2012testing}, \cite{choi2004goodness}, \cite{finkelstein1971improved} and \cite{soest1969some}, who worked in this direction and proposed their tests for testing exponentiality against various alternatives. Motivated by this work, we developed a simple non-parametric test for testing exponentiality using some behavioural properties of NDSE with respect to time. We propose a test with the idea that it should be easy to compute and have high power.

In this paper, first, we investigate the behavior of NDSE with respect to $t$ and we show that NDSE is constant only if $X$ follows an exponential distribution, it is given as Theorem \ref{NDSE:constant} in this paper. We define two classes, INDSE and DNDSE, as per the increasing and decreasing NDSE with respect to $t$. We include some interesting examples that are in INDSE and DNDSE based on the choice of their shape parameter. Next, we provide a result which states that $F$ is INDSE (or DNDSE) based on the negative (or positive) value of $\delta(x)$, which is defined in \eqref{small:delta}. Using this lemma, we further define a measure of departure for our proposed goodness-of-fit test for testing exponentiality.

The rest of the paper is organized as follows: In Section \ref{exact_test}, we proposed an exact test to investigate the exponentiality of a random sample. We define a measure of departure using $\delta(x)$, which gave our test statistic. Next, based on U-statistic, we provided a scale-invariant test statistic $\widehat{\Delta}^*$. We also provide the exact distribution of $\widehat{\Delta}^*$ in this section. Finally, we included a table containing the critical values for different sample sizes and for $99\%$ and $95\%$ confidence levels. Section \ref{Asymptotic:properties} contains asymptotic properties of $\widehat{\Delta}^*$ and an asymptotic representation of our proposed test statistic. In Section \ref{Censored:section}, we proposed a modified test statistic in the case of censored observation. We consider only the case of right-censored observations. Similar results can also be proved for left-censored. In Section \ref{Simulation:section}, we reported the simulation's result to investigate the asymptotic test's performance. Finally, we illustrated our test procedure using a real data set.
\begin{theorem}
\label{NDSE:constant}
    The NDSE is constant if and only if the distribution of $X$ is exponential distribution.
\end{theorem}
\begin{proof}
    \,\,Let us assume that NDSE $\eta(t)$ is equal to a constant $\lambda^*$. First, we show that the distribution of $X$ is exponential. Consider
    \begin{align} \lambda^*&=-\frac{1}{2}\frac{\int_{0}^{t} \bar{F}^2(u)\mathrm{d}u}{1-\bar{F}^2(t)}\nonumber\\\int_{0}^{t} \bar{F}^2(u)du&=\lambda\left (1-\bar{F}^2(t)  \right ),\label{CH_Constant_EQ[1]}\end{align}
    where $\lambda=-2\lambda^*$. Now, we solve this integral equation by differentiating both sides of the equation \eqref{CH_Constant_EQ[1]} and we get
    \begin{equation}
        \bar{F}^2(t)+\lambda\mathrm{d}\left (\bar{F}^2(t)  \right )=0,
    \end{equation}
    with $\bar{F}(0)=1$. Take $g(t)=\bar{F}^2(t)$ then we get $g(0)=1$ and an linear differential equation 
    \begin{equation}
        g(t)+\lambda g'(t)=0.
    \end{equation}
    The solution of the differential equation is 
    \begin{equation*}
        g(t)=e^{-\frac{1}{\lambda}t},
    \end{equation*}
    and then
    \begin{equation*}
        \bar{F}(t)=\pm e^{-\frac{1}{2\lambda}t}.
    \end{equation*}
    Now, using the fact that $\bar{F}(t)\geq0$, we get a unique solution of the differential equation as 
    \begin{equation}
         \bar{F}(t)= e^{-\frac{1}{2\lambda}t}.
    \end{equation}
    Next, to prove that, if $X$ follows exponential distribution with some parameter $\beta$ that is $\bar{F}(t)=e^{-\beta t}$, then $\eta(t)$ is constant. Consider 
    \begin{align*}\eta(t)&=-\frac{1}{2}\frac{\int_{0}^{t}e^{-2\beta x}\mathrm{d}x}{1-e^{-2\beta t}}\\&=\frac{1}{4\beta}\left (\frac{e^{-2\beta t}-1}{1-e^{-2\beta t}}   \right )\\&=-\frac{1}{4\beta}\end{align*}
    Hence, the proof is complete.
\end{proof}

\begin{definition}
    A random variable $X$ belongs to the increasing normalized dynamic survival extropy (INDSE) or decreasing normalized dynamic survival extropy (DNDSE) classes if the NDSE is increasing or decreasing. If $X$ follows the distribution $F(x)$, then we say $F$ is INDSE or DNDSE accordingly.
\end{definition}

Some well-known examples of INDSE are Weibull distribution, log-logistic distribution, and gamma distribution, which have a shape parameter greater than one and a positive scale parameter. The Rayleigh distribution is also INDSE. Similarly,  Weibull distribution, log-logistic distribution, and gamma distribution with a shape parameter of less than one and a positive scale parameter are DNDSE. Further, the following lemma provides a way to check whether $X$ belongs to INDSE or DNDSE. The following lemma can be easily proved using basic calculus and the definition of NDSE.
\begin{lemma}
    Let $X$ be a continuous non-negative random variable with distribution function $F(\cdot)$ and probability density function $f(\cdot)$, then $X$ belongs to INDSE (or DNDSE) class if and only if 
    \begin{equation}
    \label{small:delta}
        \delta(x)=\left ( 1-\bar{F}^2(x) \right )\bar{F}(x)-2f(x)\int_{0}^{x}\bar{F}^2(t)\mathrm{d}t< 0 (\text{or }> 0)
    \end{equation}
    for all $x>0$.
\end{lemma}

\section{An exact test for exponentiality}
\label{exact_test}
Let $X_1,\,X_2,\,\cdots,\,X_n$ be a random sample from distribution $F$. We are interested to test the hypothesis,
\begin{align*}
    H_{0} &: F \text{ is exponential}\\
    H_{1} &: F \text{ is INDSE (not exponential)}.
\end{align*}
Define the measure of departure as
\begin{equation}
    \Delta(F)=\int_{0}^{\infty} \left (\left ( 1-\bar{F}^2(x) \right )\bar{F}(x)-2f(x)\int_{0}^{x}\bar{F}^2(t)\mathrm{d}t  \right )\mathrm{d}x.
\end{equation}
This implies, $H_0 :\ \Delta(F)=0$ and $H_1:\ \Delta(F)<0$. Using Fubini's theorem, we get
\begin{align*} \Delta(F)
&=\int_{0}^{\infty}\bar{F}(x)\mathrm{d}x-\int_{0}^{\infty} \bar{F}^3(x)\mathrm{d}x-2\int_{0}^{\infty} \bar{F}^3(x)\mathrm{d}x\\
&=\int_{0}^{\infty}\bar{F}(x)\mathrm{d}x-3\int_{0}^{\infty} \bar{F}^3(x)\mathrm{d}x\\
&=\mu-3\mathbb{E}[\min{(X_1,X_2,X_3)}].
\end{align*}

\subsection{Test statistic}\label{subsec1}
A $U$-statistic based estimator of $\Delta(F)$ is given by
\begin{equation}
    \widehat{\Delta}=\binom{n}{3}^{-1}\sum_{1\leq i<j<k\leq n}h(X_i,X_j,X_k),
    \label{Ustat[1]}
\end{equation}
where,
\begin{align*}h(X_i,X_j,X_k)=\frac{1}{3}\left( X_1+X_2+X_3-9X_1I(X_1<\min{(X_2,X_3)}) -9X_2I(X_2<\min{(X_1,X_3)})\right.\\ \left.
-9X_3I(X_3<\min{(X_1,X_2)}) \right ).
\end{align*}
To make the test scale-invariant, we consider the departure measure 
\begin{equation}
   \Delta^{*}(F)=\frac{\Delta(F)}{\mu}
\end{equation}
and the test statistic is given by
\begin{equation}
    \widehat{\Delta}^{*}=\frac{\Delta(F)}{\bar{X}},
\end{equation}
where $\bar{X}=\frac{1}{n}\sum_{i=1}^{n}X_i$. Hence, the test procedure is to reject the null hypothesis $H_{0}$ in favour of the alternative hypothesis $H_1$ for smaller values of $\widehat{\Delta}^{*}$.
\begin{remark}
A similar result can be proved by considering the alternate hypothesis, as $F$ is DNDSE.
\end{remark}
\subsection{Distribution of test statistic}
We use a result due to \cite{box1954some} to find the exact null distribution of the test statistic $\widehat{\Delta}^{*}$.
\begin{theorem}
    Let $X$ be continuous non-negative random variable with $\bar{F}(x)=e^{-\frac{x}{2}}$. Let $X_1,\,X_2,\cdots,\,X_n$ be independent and identical samples from $F$. Then, for fixed $n$
    \begin{equation}
        P(\widehat{\Delta}^{*}>x)=\sum_{j=1}^{n-2}\prod_{i=1,i\ne j}^{n-1}\left ( \frac{d_{j,n}-x}{d_{j,n}-d_{i,n}} \right )I_{d_{j,n}}(x)+(1+\mu_1)\prod_{i=1}^{n-2}\left ( \frac{d_{n-1,n}-x}{d_{n-1,n}-d_{i,n}} \right )I_{d_{n-1,n}}(x)
    \end{equation}
   provided $d_{i,n}\ne d_{j,n}$ for $i\ne j$, where
    \begin{equation}
        I_y(x)=\left\{\begin{matrix}
1 &\text{ if }x\leq y&\\ 
0 &\text{ if }x> y
\end{matrix}\right.\text{and }\mu_1=\sum_{i=1}^{n-2}\left ( \frac{x-d_{i,n}}{d_{n-1,n}-d_{i,n}} \right ),
\end{equation}
and
\begin{equation}
    d_{i,n}=\frac{1}{(n-1)(n-2)}\left ( (3i-3n+6)(n-i+1)+n^2-3n-4 \right )\text{ for }i=1,2,\cdots, n-1.
\end{equation}
\end{theorem}
\begin{proof}
 We express the test statistic in terms of order statistics. Let $X_{(i)}$, $i=1,2,\cdots n$ be the $i$-th order statistic based on the random sample $X_1,\,X_2,\,\cdots,\,X_n$ from $F$. After some algebraic manipulation the Eq. \eqref{Ustat[1]} can be express as 
     \begin{align}\widehat{\Delta}&=\frac{1}{n(n-1)(n-2)}\sum_{i=1}^{n}\left ((n-1)(n-2)-9(n-i-1)(n-i)  \right )X_{(i)}\nonumber\\&=\frac{1}{n(n-1)(n-2)}\sum_{i=1}^{n}\left ( -8n^2+6n+18ni-9i-9i^2+2 \right )X_{(i)}\label{Dist_EQ[1]}\end{align}
     Rewrite the equation \eqref{Dist_EQ[1]} as
     \begin{align*}\widehat{\Delta}&=\frac{1}{n(n-1)(n-2)}\sum_{i=1}^{n}\left ( (3i-3n+6)(n-i+1)^2-\left (3(i+1)-3n+6  \right )(n-i+1)^2+n^2-3n-4 \right )X_{(i)}\\&=\frac{1}{n(n-1)(n-2)}\sum_{i=1}^{n}(3i-3n+6)(n-i+1)^2\left ( X_{(i)}-X_{(i-1)} \right )+(n^2-3n-4 )X_{(i)}\end{align*}
     Let $D_i$ denotes normalised spacing, that is, $D_i=(n-i+1)\left ( X_{(i)}-X_{(i-1)} \right )$, with $X(0)=0$, then we write
     \begin{equation}
         \widehat{\Delta}=\frac{1}{n(n-1)(n-2)}\sum_{i=1}^{n}\left ((3i-3n+6)(n-i+1)+(n^2-3n-4 )  \right )D_{i}
     \end{equation}
     Hence, we obtain
     \begin{equation}
         \widehat{\Delta}^*=\frac{\sum_{i=1}^{n}d_{i,n}D_{i}}{\sum_{i=1}^{n}D_{i}}
     \end{equation}
     where,
     \begin{equation}
         d_{i,n}=\frac{1}{(n-1)(n-2)}\left ( (3i-3n+6)(n-i+1)+n^2-3n-4 \right )\text{ for }i=1,2,\cdots, n.\label{d_i_n[1]}
     \end{equation}
     We observe that $d_{n-1,n}=d_{n,n}$ and for all $i\ne j$, $d_{i,n}\ne d_{j,n}$, $i,j\in \{1,2,3,\cdots,n-1\}$, next we define
     \begin{equation}
         G_i=\left\{\begin{matrix}
D_i, &i=1,2,\cdots,n-2 \\ 
D_{n-1}+D_n, & i=n-1.
\end{matrix}\right.
     \end{equation}
     So, we get 
     \begin{equation}
         \widehat{\Delta}^*=\frac{\sum_{i=1}^{n-1}d_{i,n}G_{i}}{\sum_{i=1}^{n-1}G_{i}},
     \end{equation}
     where $d_{i,n}$ is same as in Eq. \eqref{d_i_n[1]}. We know that the exponential distribution with rate $\frac{1}{2}$ is distributed same as $\chi^2$ random variable with 2 degrees of freedom. Hence, $G_1,G_2,\cdots G_{n-2}$ follows $\chi^2$ distribution with 2 degrees of freedom and $G_{n-1}$ follows $\chi^2$ distribution with 4 degrees of freedom. Therefore, the result follows from Theorem 2.4 of \cite{box1954some}.
\end{proof}

The critical values of the exact test can be seen in Table \ref{Critical_Values:Table}. The critical values are determined using simulation. The entire simulation in this paper is in R programming. We generate random samples from the exponential distribution using the cumulative distribution function $F(x)=1-\exp(-x),\ x\geq 0$. Since our proposed test is scale-invariant, we can take the scale parameter as unity while executing the simulation. We simulate 100000 samples of different sizes to find critical values for $99\%$ and $95\%$ confidence. We observe from the table that the critical values are increasing to zero as we increase the sample sizes $n$.
\begin{table}[t]
    \centering
    \renewcommand{\arraystretch}{1.5}
    \small
    \caption{The critical values of the exact test  at the significance level $\alpha = 0.01$ and $\alpha = 0.05$}
\vspace{15pt}
    \begin{tabular}{m{1.5cm}  m{2cm} m{2cm} | m{1.5cm} m{2cm} m{2cm}}
        \hline
        \textbf{$n$} & \textbf{$\alpha=0.01$} & \textbf{$\alpha=0.05$}&\textbf{$n$} & \textbf{$\alpha=0.01$} & \textbf{$\alpha=0.05$}\\
        \hline
  3 & $-1.71190$  & $-1.32316$ & 12 & $-0.70742$  & $-0.47744$\\
    4 & $-1.43511$  & $-1.03571$ & 13 & $-0.66480$  & $-0.45728$\\
    5 & $-1.23650$  & $-0.85287$ &14 & $-0.64394$  & $-0.43561$ \\
    6 & $-1.08402$  & $-0.75092$ & 15 & $-0.61046$  & $-0.42014$ \\
    7 & $-0.98660$  & $-0.66664$ &16 & $-0.59501$  & $-0.40037$ \\
    8 & $-0.90474$  & $-0.61714$ &17 & $-0.56900$  & $-0.38924$\\
    9 & $-0.84119$  & $-0.57126$ &18 & $-0.55223$  & $-0.37855$\\
    10 & $-0.78409$  & $-0.53411$ &19 & $-0.53952$  & $-0.36901$\\
    11 & $-0.74789$  & $-0.50425$ &20 & $-0.52161$  & $-0.35935$ \\
   \hline
    \end{tabular}
    \label{Critical_Values:Table}
\end{table}

\section{Asymptotic properties}
\label{Asymptotic:properties}
In this section, we examine the asymptotic characteristics of the suggested test statistic. The test statistic is asymptotically normal and consistent with the alternatives. The test statistic's null variance is demonstrated to be independent of the parameter.
\subsection{Consistency and asymptotic normality}
As the proposed test is based on U-statistic, we use the asymptotic theory of U-statistic to discuss the limiting behaviour of $\widehat{\Delta}^{*}$. 
\begin{theorem}
    The $\widehat{\Delta}^*$ is a consistent estimator of $\Delta^*(F)$.
\end{theorem}
\begin{proof}
    We can write 
    \begin{equation*}
        \widehat{\Delta}^{*}=\frac{\widehat{\Delta}}{\Delta(F)}\cdot \frac{\Delta(F)}{\mu}\cdot \frac{\mu}{\bar{X}}.
    \end{equation*}
    \cite{lehmann1951consistency} showed that $\widehat{\Delta}$ is a consistent estimator of $\Delta(F)$ and since $\bar{X}$ is consistent estimator of $\mu$. Therefore, $\widehat{\Delta}^*$ is a consistent estimator of $\Delta^*(F)$.
\end{proof}
Next, we find the limiting distribution of test statistic.
\begin{theorem}
\label{theorem_[var]}
    The distribution of $\sqrt{n}\left ( \widehat{\Delta}-\Delta(F) \right )$, as $n\rightarrow \infty$, is Gaussian with mean zero and variance of $9\sigma_{1}^{2}$ , where $\sigma_{1}^{2}$ is asymptotic variance of $\widehat{\Delta}$ and is given by 
    \begin{equation}
        \sigma_{1}^{2}=\frac{1}{9}Var\left ( X-9X\bar{F}^2(X)-18\int_{0}^{\infty}\left ( \int_{0}^{\min{(X,z)}} y\mathrm{dy} \right )\mathrm{dF(z)} \right ).
    \end{equation}
    \end{theorem}
    \begin{proof}
      Since the kernel has degree 3, using the central limit theorem for U-statistic, $\sqrt{n}\left ( \widehat{\Delta}-\Delta(F) \right )$ has a limiting distribution
        \begin{equation}
            N(0,9\sigma_1^2), \text{ as } \,n\rightarrow\infty,
        \end{equation}
        where $\sigma_1^2$ is specified in the theorem. For finding $\sigma_1^2$, consider 
        \begin{align*} 
            \mathbb{E}[h(x,X_{2},X_3)]&=\frac{1}{3}\mathbb{E}\left [ x+X_2+X_3-9xI(x<\min{(X_2,X_3)})-9X_2I(X_2<\min{(x,X_3)})\right.\\ &\left.-9X_3I(X_3<\min{(x,X_2)}) \right ]\\&=\frac{1}{3}\left ( x+2\mu-9x\mathbb{E}[I(x<\min{(X_2,X_3)})]-18\mathbb{E}[X_2I(X_2<\min{(x,X_3)})] \right )\\&=\frac{1}{3}
            \left ( x+2\mu-9x\bar{F}^2(x)-18\int_{0}^{\infty}\left ( \int_{0}^{\min{(x,z)}}y\mathrm{dy} \right )\mathrm{dF(z)} \right ).
        \end{align*}
    Hence,
    \begin{equation}
        \sigma_{1}^{2}=\frac{1}{9}Var\left ( X-9X\bar{F}^2(X)-18\int_{0}^{\infty}\left ( \int_{0}^{\min{(X,z)}} y\mathrm{dy} \right )\mathrm{dF(z)} \right ).
    \end{equation}
    which completes the proof.
\end{proof}
\begin{corollary}
    Let $X$ be continuous non-negative random variable with $\bar{F}(x)=e^{-\frac{x}{\lambda}}$, then $\sqrt{n}\, \widehat{\Delta}$, as $n\rightarrow \infty$, is Gaussian random variable with mean zero and variance $\sigma_{0}^{2}=\frac{4}{5}\lambda^2$.
\end{corollary}
\begin{proof}
    Under $H_0$, we have 
    \begin{equation}
        \int_{0}^{\infty}\left ( \int_{0}^{\min{(x,z)}} y\mathrm{dy} \right )\mathrm{dF(z)}=-\frac{1}{2}xe^{-\frac{2x}{\lambda}}-\frac{1}{4}\lambda e^{-\frac{2x}{\lambda}}+\frac{1}{4}\lambda.
    \end{equation}
    Now using Theorem \ref{theorem_[var]}, 
    \begin{align*}
        \sigma_0^2&=Var\left ( X+\frac{9}{2}\lambda e^{-\frac{2X}{\lambda}} -\frac{9}{2}\lambda\right )\\
        &=Var\left ( X+\frac{9}{2}\lambda e^{-\frac{2X}{\lambda}}\right )\\
        &=Var(X)+\frac{81\lambda^2}{4}Var(e^{-\frac{2X}{\lambda}})+9\lambda Cov(X,e^{-\frac{2X}{\lambda}})\\
        &=\lambda^2+\frac{9\lambda^2}{5}-2\lambda^2\\
        &=\frac{4}{5}\lambda^2.
    \end{align*}
\end{proof}
Now, using Slutsky's theorem, the following result can be obtained using the above corollary.
\begin{corollary}
\label{cor:assymtotic}
    Let $X$ be continuous non-negative random variable with $\bar{F}(x)=e^{-\frac{x}{\lambda}}$, then $\sqrt{n}\, \widehat{\Delta}^*$, as $n\rightarrow \infty$, is Gaussian random variable with mean zero and variance $\sigma_{0}^{2}=\frac{4}{5}$.
\end{corollary}
Hence, in the case of the asymptotic test, for small values of $n$, we reject null hypothesis $H_0$ in favor of the alternative hypothesis $H_1$ if 
\begin{equation*}
    \sqrt{\frac{5}{4}n}\,\,\, \widehat{\Delta}^*<-z_{\alpha},
\end{equation*}
where $z_\alpha$ is the upper $\alpha-$percentile of $N(0,1)$.
\begin{remark}
    Similarly, we reject the null hypothesis $H_0$ in favor of the alternative hypothesis $H_1$ that is, $F$ is DNDSE, if 
\begin{equation*}
    \sqrt{\frac{5}{4}n}\,\,\, \widehat{\Delta}^*>z_{\alpha},
\end{equation*}
where $z_\alpha$ is the upper $\alpha-$percentile of $N(0,1)$.
\end{remark}

\section{The case of censored observations}
\label{Censored:section}
    Occurrences of right-censored observations are frequently seen in the analysis of lifetime data. Various techniques address the issue of testing for exponentiality using censored samples. An alternative method involves replacing the distribution function with the Kaplan-Meier estimator in order to calculate the test statistic. In this technique, it is necessary to modify the metric used to quantify deviation from the null hypothesis in the presence of censored observations.
    Another method is the inverse probability censoring weighted scheme (IPCW), in which the censored data is adjusted by weighting it with the inverse of the survival function of the censoring variable provided by (\cite{koul1981regression}, \cite{rotnitzky2005inverse} and \cite{datta2010inverse}). In this discussion, we explore the approach to address instances of censorship.

Assume that we have randomly censored observations, meaning that the censoring times are unrelated to the lifetimes and occur independently. Let the observed data are $n$ independent and identical (i.i.d.) copies of $(X^*,\delta)$, with $X^*=min(X,C)$, where $C$ is the censoring time and $\delta=I(X\leq C)$. We investigate the testing problem mentioned in Section \ref{exact_test} based on $n$ i.i.d. observations $\{(X_i,\delta_i),\, 1\leq i\leq n\}$. Note that $\delta_i =0$ means that the $i$th object is censored by $C$, on the right and $\delta_i =1$ means $i$th object is not censored. We refer to \cite{koul1980testing} to define measure $\Delta(F)$ for censored observations. We refer to \cite{datta2010inverse} to get an estimator of $\Delta(F)$ with censored observation as 
\begin{equation}
     \widehat{\Delta}_c =\frac{6}{n(n-1)(n-2)}\sum_{1\leq i<j<k\leq n} \frac{h(X_i^*,X_j^*,X_k^*)\delta_i \delta_j \delta_k}{\widehat{K}_c(X_i^*) \widehat{K}_c(X_j^*)\widehat{K}_c(X_k^*)},
\end{equation}
where $\widehat{K}_c(X_i^*), \widehat{K}_c(X_j^*),\widehat{K}_c(X_k^*)$ are strictly positive with probability one and 
\begin{align*}h(X_1^*,X_2^*,X_3^*)=\frac{1}{3}\left( X_1^*+X_2^*+X_3^*-9X_1^*I(X_1^*<\min{(X_2^*,X_3^*)}) -9X_2^*I(X_2^*<\min{(X_1^*,X_3^*)})\right.\\ \left.
-9X_3^*I(X_3^*<\min{(X_1^*,X_2^*)}) \right ).
\end{align*}
Here, $\widehat{K}_c$ is the Kaplan-Meier estimator of $K_c$, the survival function of the censoring variable $C$. In the same way, an estimator of $\mu$ is given by
\begin{equation}
    \widehat{X}_c=\frac{1}{n}\sum_{1}^{n}\frac{X_i^* \delta_i}{\widehat{K}_c(X_i^*)}.
\end{equation}
Therefore, in the right censoring situation, the test statistic is given by
\begin{equation}
    \widehat{\Delta}_c^*=\frac{\widehat{\Delta}_c}{\widehat{X}_c},
\end{equation}
and the test procedure is to reject null hypothesis $H_0$ in favor of $H_1$ for small values of $\widehat{\Delta}_c^*$.

For deriving the asymptotic distribution of $\widehat{\Delta}_c^*$, let us define $N_i^c(t) = I(X_i^* \le t, \delta_i = 0)$ as the counting process corresponding to the censoring random variable for the $i$-th subject and $R_i(u) = I(X_i^* \ge u)$. Let $\lambda_c(t)$ be the hazard rate of the censoring variable $C$. The martingale associated with the counting process \(N_i^c(t)\) is given by
\begin{equation}
M_i^c(t) = N_i^c(t) - \int_0^t R_i(u) \lambda_c(u) \, \mathrm{d}u.
\end{equation}

Let $G(x,y)=P(X_1\leq x, X_1^*\leq y, \delta_1=1)$, $x\in \mathbb{R}$, $H(t)=P(X_1^*\geq t)$ and 
\begin{equation}
    w(t) = \frac{1}{H(t)} \int_{\mathbb{R}\times [0,\infty)} \frac{h_1(x)}{K_c(y -)} I(y > t) \, \mathrm{d}G(x,y),
\end{equation}
where $h_1(x)=\mathbb{E}h(x,X_2^*,X_3^*)$. The next theorem follows from \cite{datta2010inverse} for the choice of the kernel 
\begin{align*}h(X_1^*,X_2^*,X_3^*)=\frac{1}{3}\left( X_1^*+X_2^*+X_3^*-9X_1^*I(X_1^*<\min{(X_2^*,X_3^*)}) -9X_2^*I(X_2^*<\min{(X_1^*,X_3^*)})\right.\\ \left.
-9X_3^*I(X_3^*<\min{(X_1^*,X_2^*)}) \right ),
\end{align*}
and under the assumption $\mathbb{E}h^2(X_1^*,X_2^*,X_3^*)<\infty$, 
\begin{equation*}
    \int_{\mathbb{R}\times [0,\infty)} \frac{h_1^2(x)}{K_c^2(y)}\, \mathrm{d}G(x,y)<\infty,
\end{equation*}
and 
\begin{eqnarray*}
    \int_{0}^{\infty} w^2(t)\lambda_c(t)\mathrm{d}t<\infty.
\end{eqnarray*}
\begin{theorem}
\label{theorem:censor}
    The distribution of $\sqrt{n}\left(\widehat{\Delta}_c-\Delta(F)\right)$, as $n\rightarrow \infty$, is Gaussian with mean zero and variance $9\sigma_{1c}^2$, where $\sigma_{1c}^2$ is given by
    \begin{eqnarray*}
        \sigma_{1c}^2=Var\left( \frac{h_1(X)\delta_1}{K_c(X^*)}+\int w(t)\,\mathrm{d}M_1^c(t)\right).
    \end{eqnarray*}
\end{theorem}
\begin{corollary}
    Under the assumption of Theorem \ref{theorem:censor}, if $\mathbb{E}(X_1^2)<\infty$, the distribution of  $\sqrt{n}\left(\widehat{\Delta}_c^*-\Delta^*\right)$, as $n\rightarrow \infty$, is Gaussian with mean zero and variance $9\sigma_{c}^2$, where $\sigma_{c}^2$ is given by
    \begin{equation}
        \sigma_c^2=\frac{\sigma_{1c}^2}{\mu^2}.
    \end{equation}
\end{corollary}
\begin{proof}
    The consistency of the estimator $\widehat{X}_c$ for $\mu$ is proved by \cite{zhao2000estimating}. Therefore, the result follows from the above theorem by applying Slutsky's theorem. 
\end{proof}
\begin{corollary}
    Let $X$ be continuous non-negative random variable with $\bar{F}(x)=e^{-\frac{x}{\lambda}}$. Under the assumption of Theorem \ref{theorem:censor}, if $\mathbb{E}(X^2)<\infty$, the distribution of  $\sqrt{n}\, \widehat{\Delta}_c^*$, as $n\rightarrow \infty$, is Gaussian random variable with mean zero and variance $\sigma_{c0}^{2}$, where
    \begin{equation}
    \label{eqn:censoring:var}
        \sigma_{c0}^{2}=\frac{9}{\lambda^2}Var\left(\frac{(X+\frac{9}{2}\lambda\bar{F}^2(X))\delta_1}{3K_c(X^*-)}+\int w(t)\,\mathrm{d}M_1^c(t)\right).
    \end{equation}
\end{corollary}
Note that, if all the observations are uncensored, then the second term in the variance expression \eqref{eqn:censoring:var} will be zero and the variance reduces to the null variance given in Corollary \ref{cor:assymtotic}. Therefore, in the case of right censoring, we reject null hypothesis $H_0$ in favor of $H_1$, if
\begin{equation}
    \frac{\sqrt{n}\widehat{\Delta}_c^*}{\widehat{\sigma}_{c0}}< -z_\alpha,
\end{equation}
where $\widehat{\sigma}_{c0}$ is a consistent estimator of $\sigma_{c0}$ and can be estimated using \eqref{eqn:censoring:var}.
\section{Simulation and data analysis}
\label{Simulation:section}
We calculate the power of our exact test using the Monte Carlo method by taking alternate distributions such as Weibull, gamma and lognormal distribution. We also compare the power of our proposed test with some other known tests for testing exponentiality. All the tables representing power for some alternative testing procedures are also included in this section.

\subsection{Power calculation}

\cite{zardasht2015empirical} provide a goodness-of-fit test for exponential distribution, where the test statistic is based on the well-known estimate of $\bar{F}$, Kaplan-Meier estimator. They reject $H_0$ in favor of $H_1$ at significant level $\alpha$ if 
\begin{equation}
    \sqrt{\frac{382n}{5}}|C_n-0.25|>z_{1-\frac{\alpha}{2}},
\end{equation}
where 
\begin{eqnarray*}
    C_n=\frac{1}{n}\sum_{i=1}^{n} \frac{X_i}{\bar{X}}e^{-\frac{X_i}{\bar{X}}}.
\end{eqnarray*}
\cite{baratpour2012testing} provide a test statistic for assessing the goodness-of-fit. This statistic is derived from a discrimination measure that is based on a modified version of the Kulback-Leibler information measure. This test aims to compare the hypothesis $H_0$ with the alternative hypothesis $H_1$. The test statistic is given by 
\begin{equation}
    T_n = \frac{\sum_{i=1}^{n-1}  \frac{n-i}{n} \left(\ln  \frac{n-i}{n} \right) \left( X_{(i+1)} - X_{(i)} \right)  + \frac{\sum_{i=1}^{n} X_i^2}{2 \sum_{i=1}^{n} X_i}}{\frac{\sum_{i=1}^{n} X_i^2}{2 \sum_{i=1}^{n} X_i}},
\end{equation}
where $X_{(i)}$ is the $i$th order statistic related to the sample and $H_0$ is rejected at significant level $\alpha$ if $T_n\geq T_{n,1-\alpha}$, where $T_{n,1-\alpha}$ is the $100(1-\alpha)$ percentile of $T_n$ under $H_0$. In addition, they conduct a Monte Carlo simulation analysis to compare the performance of $T_n$ with the statistic established by \cite{soest1969some}
\begin{equation}
    W^2 =  \sum_{i=1}^{n} \left[ F_0(X_{(i)}, \hat{\lambda}) - \frac{2i - 1}{2n} \right]^2 + \frac{1}{12n},
\end{equation}
the statistic proposed by \cite{finkelstein1971improved}
\begin{equation}
    S^* = \sum_{i=1}^{n}\max\left\{ \left| F_0(X_{(i)}, \hat{\lambda}) - \frac{i}{n} \right|, \left| F_0(X_{(i)}, \hat{\lambda}) - \frac{i - 1}{n} \right| \right\},
\end{equation}
where $F_0(x,\lambda)=1-e^{\frac{-x}{\lambda}}$, $\hat{\lambda}=\bar{X}=\frac{1}{n}\sum_{i=1}^{n}X_{i}$ and the statistic introduced by \cite{choi2004goodness}
\begin{equation}
    KLC_{mn} = \frac{\exp(C_{mn})}{\exp(\ln \bar{X} + 1)},
\end{equation}
where 
\begin{equation}
    C_{mn} = -\frac{1}{n} \sum_{i=1}^{n} \log \left( \frac{\sum_{j=i-m}^{i+m} (X_{(j)} - X_{(i)})(j-i)}{n \sum_{j=i-m}^{i+m} (X_{(j)} - \bar{X}_i)^2}\right),
\end{equation}
and 
\begin{equation}
    \bar{X}_i = \frac{1}{2m+1} \sum_{j=i-m}^{i+m} X_{(j)},
\end{equation}
which are proposed for testing $H_0$ against $H_1$. In $KLC_{mn}$ statistic, the window size $m$ is a positive integer less than $\frac{n}{2}$, $X_{(j)}=X_{(1)}$, if $j<1$, and $X_{(j)}=X_{(n)}$, if $j>n$. The null hypothesis $H_0$ is rejected for large values of $W^2$, $S^*$ and small values of $KLC_{mn}$. For simulation, we consider the following distribution functions and compare the empirical power of our test for each distribution.

\begin{table}[ht]
    \centering
    \renewcommand{\arraystretch}{1.3}
    \small
    \caption{Power comparison between the tests $\widehat{\Delta}^*$, $C_n$, $T_n$, $W^2$, $L^*$, $KLC_{mn}$ at the significance level
$\alpha = 0.05$, when the alternative distribution is Weibull }
\vspace{15pt}
    \begin{tabular}{m{1.5cm}  m{1.5cm} m{1.5cm} m{1.5cm} m{1.5cm} m{1.5cm} m{1.5cm} m{1.5cm}}
        \hline
        \textbf{$n$} & \textbf{$\beta$} & \textbf{$\widehat{\Delta}^*$} & \textbf{$C_n$} & \textbf{$T_n$} & \textbf{$W^2$} & \textbf{$S^*$} & \textbf{$KLC_{mn}$} \\
        \hline
  5 & 1 & 0.0520 & 0.0422 & 0.0502 & 0.0497 & 0.0488 & 0.0499\\
    & 2 &0.3768 &  0.2363& 0.3456& 0.2916& 0.2987& 0.3682\\
    & 3 &0.7330 & 0.5731& 0.7071& 0.6429& 0.6586& 0.7324\\
    & 4 & 0.9078 & 0.8188& 0.9106& 0.8663& 0.8801& 0.9190\\
  10& 1 &0.0542 &  0.0430& 0.0515& 0.0502& 0.0498& 0.0513\\
    & 2 &0.7552 &  0.6145& 0.6557& 0.6099& 0.6260& 0.7557\\
    & 3 &0.9885 & 0.9696& 0.9799& 0.9677& 0.9752& 0.9930\\
    & 4 &0.9997 & 0.9990& 0.9997& 0.9991& 0.9994& 0.9999\\
  15& 1 &0.0499 & 0.0451& 0.0502& 0.0502& 0.0502& 0.0505\\
    & 2 &0.9260 & 0.8446& 0.8259& 0.8244& 0.8421& 0.9216\\
    & 3 &0.9999 & 0.9988& 0.9991& 0.9986& 0.9990& 0.9998\\
    & 4 &1.0000 & 1.0000& 1.0000& 1.0000& 1.0000& 1.0000\\
  20& 1 &0.0472 & 0.0448& 0.0506& 0.0530& 0.0520& 0.0524\\
    & 2 &0.9785 & 0.9429& 0.9175& 0.9326& 0.9404& 0.9790\\
    & 3 &1.0000& 0.9999& 0.9999& 0.9999& 0.9999& 1.0000\\
    & 4 &1.0000& 1.0000& 1.0000& 1.0000& 1.0000& 1.0000\\
  25& 1 &0.0523 & 0.0461& 0.0518& 0.0484& 0.0502& 0.0510\\
    & 2 &0.9948 & 0.9827& 0.9644& 0.9759& 0.9807& 0.9940\\
    & 3 &1.0000 & 1.0000& 1.0000& 1.0000& 1.0000& 1.0000\\
    & 4 &1.0000& 1.0000& 1.0000& 1.0000& 1.0000& 1.0000\\
   \hline
    \end{tabular}
    \label{Table:Weibull}
\end{table}

We consider the probability density function for the Weibull distribution  as
\[
f(x; \lambda, \beta) = \frac{\beta}{\lambda} \left( \frac{x}{\lambda} \right)^{\beta-1} e^{-(\frac{x}{\lambda})^\beta}, \hspace{0.1cm} x > 0,\hspace{0.1cm} \beta>0, \hspace{0.1cm}\lambda>0,
\]
for the gamma distribution as
\[
f(x; \lambda, \beta) = \frac{1}{\Gamma(\beta)\lambda^\beta} x^{\beta-1} e^{-\frac{x}{\lambda}}, \hspace{0.1cm}  x > 0,\hspace{0.1cm} \beta>0, \hspace{0.1cm}\lambda>0,
\]
and for the lognormal distribution as
\[
f(x; \mu, \sigma^2) = \frac{1}{x \sigma \sqrt{2\pi}} \exp\left( -\frac{(\ln x - \mu)^2}{2\sigma^2} \right), \hspace{0.1cm} x > 0,\hspace{0.1cm} -\infty<\mu<\infty,\  \sigma>0.
\]

We refer \cite{baratpour2012testing} for generating  samples, they set the parameters for each distribution such that $$\frac{E(X_1^2)}{2E(X_1)} = 1,$$ that is,
\begin{itemize}
    \item For the Weibull distribution:
    \[
    \lambda = \frac{2 \Gamma(1 + \frac{1}{\beta})}{\Gamma(1 + \frac{2}{\beta})}
    \]
    \item For the gamma distribution:
    \[
    \lambda = \frac{2}{1 + \beta}
    \]
    \item For the lognormal distribution:
    \[
    \sigma^2 = \frac{2}{3} (\ln 2 - \mu)
    \]
\end{itemize}

We calculate the empirical power for each statistic using $10000$ produced samples with sample sizes of $n = 5,\, 10,\, 15,\, 20$ and $25$. We calculate the power as the ratio of the number of times the associated statistic exceeded the relevant threshold to 10000. The results are summarized in Tables \ref{Table:Weibull}, \ref{Table:Gamma} and \ref{Table:lognormal}.
\begin{table}[t]
    \centering
    \renewcommand{\arraystretch}{1.3}
    \small
    \caption{Power comparison between the tests $\widehat{\Delta}^*$, $C_n$, $T_n$, $W^2$, $L^*$, $KLC_{mn}$ at the significance level
$\alpha = 0.05$, when the alternative distribution is gamma}
\vspace{15pt}
    \begin{tabular}{m{1.5cm}  m{1.5cm} m{1.5cm} m{1.5cm} m{1.5cm} m{1.5cm} m{1.5cm} m{1.5cm}}
        \hline
        \textbf{$n$} & \textbf{$\beta$} & \textbf{$\widehat{\Delta}^*$} & \textbf{$C_n$} & \textbf{$T_n$} & \textbf{$W^2$} & \textbf{$S^*$} & \textbf{$KLC_{mn}$} \\
        \hline
           5 & 1 &0.0529 & 0.0419 & 0.0501 & 0.0501 & 0.0494 & 0.0493 \\
         & 5 &0.6194& 0.4079 & 0.5162 & 0.4874 & 0.4886 & 0.5696 \\
        & 6 & 0.6989 & 0.5030 & 0.6058 & 0.5848 & 0.5841 & 0.6662\\
         & 7 & 0.7831 &0.5878 & 0.6826 & 0.6676 & 0.6654 & 0.7443\\
        10 & 1 &0.0471 &  0.0436 & 0.0490 & 0.0499 & 0.0506 & 0.0497  \\
        & 5 &0.9656 & 0.9153 & 0.8334 & 0.8927 & 0.8901 & 0.9403\\
        & 6 &0.9907 & 0.9686 & 0.9096 & 0.9547 & 0.9526 & 0.9777\\
        & 7 & 0.9977 &0.9896 & 0.9535 & 0.9828 & 0.9803 & 0.9919\\
         15 & 1 &0.0505 & 0.0432 & 0.0510 & 0.0516 & 0.0500 & 0.0504\\
        & 5 &0.9981 & 0.9940 & 0.9462 & 0.9883 & 0.9862 & 0.9929 \\
        & 6 & 1.0000 &0.9994 & 0.9823 & 0.9983 & 0.9977 & 0.9986\\
        & 7 & 1.0000 &0.9999 & 0.9939 & 0.9996 & 0.9995 & 0.9997\\
         20 & 1 & 0.0489 &0.0424 & 0.0513 & 0.0492 & 0.0485 & 0.0506 \\
        & 5 &1.0000 & 0.9997 & 0.9836 & 0.9989 & 0.9987 & 0.9987\\
        & 6 & 1.0000 &1.0000 & 0.9966 & 1.0000 & 0.9999 & 0.9998 \\
        
        & 7 &1.0000 & 1.0000 & 0.9995 & 1.0000 & 1.0000 & 1.0000\\
         25 & 1 &0.0513 & 0.0462 & 0.0501 & 0.0505 & 0.0512 & 0.0497 \\
        & 5 & 1.0000 &1.0000 & 0.9955 & 0.9999 & 0.9999 & 0.9997\\
        & 6 &1.0000 & 1.0000 & 0.9994 & 1.0000 & 1.0000 & 1.0000\\
         & 7 & 1.0000& 1.0000 & 1.0000 & 1.0000 & 1.0000 & 1.0000\\
        \hline
    \end{tabular}
    \label{Table:Gamma}
\end{table}
As we know, Weibull and gamma distributions reduce to exponential distribution for $\beta=1$. We observe that the power values for $\beta=1$ in Table \ref{Table:Weibull} and Table \ref{Table:Gamma}, is in fact the simulated level, which is close to the significance level $0.05$. One can see from the table that the power of all tests against any alternative increases with the sample size, which shows the consistency of the tests. We observe that our proposed test performs better than the alternative test statistics $C_n$, $T_n$, $W^2$, $L^*$, we also observe that our test power is similar to $KLC_{mn}$,  but our test has an added advantage of having a simple form and a known asymptotic distribution. We also have an exact distribution of our test statistic, which provides more attention to our test.

\begin{table}[t]
    \centering
    \renewcommand{\arraystretch}{1.3}
    \small
    \caption{Power comparison between the tests $\widehat{\Delta}^*$, $C_n$, $T_n$, $W^2$, $L^*$, $KLC_{mn}$ at the significance level
$\alpha = 0.05$, when the alternative distribution is lognormal}
\vspace{15pt}
    \begin{tabular}{m{1.5cm}  m{1.5cm} m{1.5cm} m{1.5cm} m{1.5cm} m{1.5cm} m{1.5cm} m{1.5cm}}
        \hline
        \textbf{$n$} & \textbf{$\mu$} & \textbf{$\widehat{\Delta}^*$} & \textbf{$C_n$} & \textbf{$T_n$} & \textbf{$W^2$} & \textbf{$S^*$} & \textbf{$KLC_{mn}$} \\
        \hline
           5 & 0.4 & 0.4392 & 0.6724 & 0.5203 & 0.5232 & 0.5150 & 0.6022 \\
     & 0.5 & 0.6592 & 0.8568 & 0.7136 & 0.7365 & 0.7203 & 0.8022 \\
     & 0.6 & 0.9478 & 0.9923 & 0.9478 & 0.9699 & 0.9588 & 0.9828 \\
        10 & 0.4 & 0.9327 & 0.9750 & 0.7962 & 0.9248 & 0.9002 & 0.9328 \\
         & 0.5 & 0.9959 & 0.9994 & 0.9497 & 0.9950 & 0.9896 & 0.9932 \\
         & 0.6 & 1.0000 & 1.0000 & 0.9997 & 1.0000 & 1.0000 & 1.0000 \\
        15 & 0.4 & 0.9959 & 0.9984 & 0.9046 & 0.9939 & 0.9899 & 0.9851 \\
         & 0.5 & 1.0000 & 0.9993 & 0.9905 & 1.0000 & 0.9999 & 0.9994 \\
     & 0.6 & 1.0000 & 1.0000 & 1.0000 & 1.0000 & 1.0000 & 1.0000 \\
        20 & 0.4 & 0.9999 & 0.9990 & 0.9557 & 0.9997 & 0.9992 & 0.9932 \\
         & 0.5 & 1.0000 & 1.0000 & 0.9983 & 1.0000 & 1.0000 & 0.9999 \\
         & 0.6 & 1.0000 & 1.0000 & 1.0000 & 1.0000 & 1.0000 & 1.0000 \\
        25 & 0.4 & 1.0000 & 0.9990 & 0.9791 & 1.0000 & 1.0000 & 0.9970 \\
         & 0.5 & 1.0000 & 1.0000 & 0.9997 & 1.0000 & 1.0000 & 1.0000 \\
         & 0.6 & 1.0000 & 1.0000 & 1.0000 & 1.0000 & 1.0000 & 1.0000 \\
        \hline
    \end{tabular}
    \label{Table:lognormal}
\end{table}
\subsection{Data Analysis}
In this part, we provide two numerical instances using real-life data to demonstrate the application of the test statistic $\widehat{\Delta}^*$ for evaluating the suitability of an exponential distribution for fitting a given data set.

\begin{example}
    In a research conducted by \cite{proschan1963theoretical}, data was collected on the time, measured in hours of operation, between consecutive failures of air-conditioning equipment in 13 aircraft. The purpose of the study was to investigate the aging properties of the equipment. The data about plane number 3 is as follows:

 90, 10, 60, 186, 61, 49, 14, 24, 56, 20, 79, 84, 44, 59, 29, 118, 25, 156, 310, 76, 26, 44, 23, 62, 130, 208, 70, 101, 208.
\end{example}
We calculate the value of test statistic $\widehat{\Delta}^*$ and it gives $\widehat{\Delta}^*=-0.2352$, which is greater than the critical value for a sample with 29 points, that is, $-0.2891$. The standard normal approximation of $\sqrt{\frac{(5)(29)}{4}}\, \widehat{\Delta}^*$ gives the P-value $0.0783$. Thus, the test does not reject the null hypothesis that the failure time follows an exponential distribution at significance level $\alpha=0.05$. \cite{lawless2011statistical} also obtained the same result for the above data using some other tests.

\begin{example}
    This dataset provides the failure points, measured in thousands of kilometers, of various locomotive controls in a life-testing experiment that included 96 controls. The experiment was concluded after covering a distance of 135000 miles, and a total of 37 failures were recorded. The failure times of the malfunctioning devices are expressed in thousands of miles. We refer to \cite{lawless2011statistical} to get the following data.
    
\[
\begin{array}{cccccc}
22.5 & 37.5 & 46.0 & 48.5 & 51.5 & 53.0 \\
54.5 & 57.5 & 66.5 & 68.0 & 69.5 & 76.5 \\
77.0 & 78.5 & 80.0 & 81.5 & 82.0 & 83.0 \\
84.0 & 91.5 & 93.5 & 102.5 & 107.0 & 108.5 \\
112.5 & 113.5 & 116.0 & 117.0 & 118.5 & 119.0 \\
120.0 & 122.5 & 123.0 & 127.5 & 131.0 & 132.5 \\
134.0 \\
\end{array}
\]\end{example}

We calculate the value of test statistic $\widehat{\Delta}^*$ and it gives $\widehat{\Delta}^*=-1.0941$, which is very less than the critical value for a sample with 37 points, that is, $-0.2554$. The standard normal approximation of $\sqrt{\frac{5*37}{4}}\widehat{\Delta}^*$ gives the P-value $0.00005\times10^{-10}$ at a significance level $\alpha=0.05$. Thus, the test rejects the null hypothesis that the data follows an exponential distribution. \cite{lawless2011statistical} and \cite{dube2011parameter} demonstrated that the lognormal distribution is a very excellent fit for this data set. This validates our result, too. We also test the same with some other alternative tests given by \cite{baratpour2012testing} and \cite{zardasht2015empirical}, their test also reject the null hypothesis $H_0$ at a significance level $\alpha=0.05$.

\section{Conclusion}
We introduced a new dynamic version of CRJ as NDSE and using this, we demonstrate a simple goodness-of-fit test for testing exponentiality. We evaluate an exact null distribution of our proposed test statistic. We derive its asymptotic properties also. We also study how the proposed method considers censoring information. We extensively compare our test with other mentioned tests. Finally, some real-life examples are presented and the test procedure is illustrated using real data. 

\section{Conflict of interest}
No conflicts of interest are disclosed by the authors.

\section{Funding}
Gaurav Kandpal would like to acknowledge financial support from the University Grant Commission, Government of India (Student ID : 221610071232).

\end{document}